\documentclass[english]{ourlematema}
\usepackage{xcolor}

\title{Uniqueness of MHV gravity amplitudes}
\titlemark{Uniqueness of MHV gravity amplitudes}


\author{J. Koefler}
\address{%
Non-linear algebra\\
Max-Planck Institute\\
\email{joris.koefler@mis.mpg.de}
}
\author{U. Oktem}
\address{%
Center for Quantum Mathematics and Physics\\
University of California, Davis\\
\email{ucoktem@ucdavis.edu}
}
\author{S. Paranjape}
\address{%
Center for Quantum Mathematics and Physics\\
University of California, Davis\\
and\\
Department of Physics, Brown University\\
\email{shruti\_paranjape@brown.edu}
}
\author{J. Trnka}
\address{%
Center for Quantum Mathematics and Physics\\
University of California, Davis\\
\email{trnka@ucdavis.edu}
}
\author{B. Zacovic}
\address{%
Department of Mathematics\\
University of Michigan, Ann Arbor\\
\email{bzacovic@umich.edu}
}

\date{2024/10/15}

\DeclareMathOperator\gr{Gr}
\DeclareMathOperator\Z{\mathbb{Z}}

\DeclareMathOperator\C{\mathbb{C}}

\newcommand{\ang}[1]{\left< #1\right>}
\newcommand{\sq}[1]{\left[ #1\right]}

\usepackage{amsmath}

\begin{document}

\maketitle

\begin{abstract}
\noindent We investigate MHV tree-level gravity amplitudes as defined on the spinor- helicity variety. Unlike their gluon counterparts, the gravity amplitudes do not have logarithmic singularities and do not admit Amplituhedron-like construction. Importantly, they are not determined just by their singularities, but rather their numerators have interesting zeroes. We make a conjecture about the uniqueness of the numerator and explore this feature from a more mathematical perspective. This leads us to a new approach for examining adjoints. We outline steps of our proposed proof and provide computational evidence for its validity in specific cases. 

\end{abstract}

\section{Introduction}\label{sec:intro}

Scattering amplitudes are mathematical functions that describe probabilities of elementary particle interactions. In the textbook formulation of quantum field theory, they are calculated as a sum of Feynman diagrams which provide a diagrammatic method to organize perturbative calculations. While this approach is general and can be used to calculate amplitudes in any quantum field theory, the efficiency of this method is problematically low for scattering of particles with spin. This is due to the number of Feynman diagrams, representing complicated rational functions, growing extremely fast with increasing multiplicity. In the last few decades, it has become clear that Feynman's picture often hides many surprising properties: extraordinary simplicity, hidden symmetries and fascinating connections to mathematical structures. 

In the context of gluon interactions, the simplest tree-level Maximal-\break Helicity-Violating (MHV) amplitude for an $n-$particle interaction is given by the famous Parke-Taylor formula \cite{Parke:1986gb},
\begin{equation} \label{PT}
    {\rm PT}_n = \frac{1}{\langle 12\rangle\langle 23\rangle\langle 34\rangle\dots\langle n1\rangle}\,.
\end{equation}
We can associate this function with the canonical form on the non-negative Grassmannian $\gr_{\geq}(2,n)$, which is the subset of the real Grassmannian $\gr(2,n)$ where all ordered maximal minors are positive \cite{Arkani-Hamed:2009ljj}. Generalizing this leads to a connection between cells in $\gr_{\geq}(k,n)$, plabic graphs and on-shell diagrams \cite{Arkani-Hamed:2012zlh} which are terms in the BCFW recursion relations for scattering amplitudes \cite{Britto:2005fq}. This further generalizes to all tree-level amplitudes and loop integrands in planar ${\cal N}=4$ super Yang-Mills (SYM) theory in the context of the Amplituhedron \cite{Arkani-Hamed:2013jha}. However, no such picture is known for amplitudes of gravitons. While the color-kinematics duality \cite{Bern:2008qj} suggests a connection between graviton and gluon scattering, the geometric picture for graviton amplitudes and connections to mathematics have not been found yet, despite some promising avenues \cite{Trnka:2020dxl,Paranjape:2023qsq}.

In this paper, we focus on the simplest MHV tree-level graviton amplitudes. There are several representations of this amplitude in the literature \cite{Berends:1988zp,Mason:2009afn,Bern:1998sv,Nguyen:2009jk}, but one of particular interest to us is \emph{Hodges formula}, Equation \eqref{eq:Hodges}, which first appeared in \cite{Hodges:2012ym}. Let $n\geq 5$ denote the number of particles involved in the scattering process, and consider the symmetric $n\times n$ matrix $(\Phi_{ij})$ for $1\leq i<j\leq n$ with entries given by 
\begin{align}\label{eq:phi_hodges}
\Phi_{ij}=\begin{cases}
    \frac{\sq{ij}}{\ang{ij}}, & i\neq j \\
    -\sum_{k\in\{1,\ldots,n\}\setminus i}\frac{\sq{ik}\ang{xk}\ang{yk}}{\ang{ik}\ang{xi}\ang{yi}}, & i=j
\end{cases},
\end{align}
where the $\ang{ij}$, $\sq{ij}$ are variables called \emph{spinors}, and $x,y\in \{1,\ldots,n\}$ are referred to as \emph{reference spinors} (see Section \ref{sec:math} for details).
Note that the spinors are antisymmetric with respect to the labels used, that is $\ang{ij}=-\ang{ji}$. 
Denote $\Phi^{\{i_1,\dots,i_k\}}_{\{j_1,\dots,j_m\}}$ the matrix obtained from $\Phi$ by deleting the rows $i_1,\dots,i_k$ and columns $j_1,\dots,j_m.$
Then, the \emph{${\rm MHV}$ gravity amplitude} $A_n$ (with the helicity factor stripped-off) is defined as the rational function
\begin{equation}\label{eq:Hodges}
A_n = \frac{\det\Phi^\mathcal{R}_\mathcal{C}}{(\mathcal{R})(\mathcal{C})},
\end{equation}
where $\mathcal{R}=\{a<b<c\}$, $\mathcal{C}=\{d<e<f\}$ are subsets of $\{1,\ldots, n\}$, with $(\mathcal{R})$ and $(\mathcal{C})$ equal $-\ang{ab}\ang{bc}\ang{ac}
$ and $-\sq{de}\sq{ef}\sq{df}$, respectively. We will later, in Section \ref{sec:math}, show that for any choice of reference spinors $x,y$, deleted rows $\mathcal{R}$ and columns $\mathcal{C}$, the numerators $N_n$ of the rational function $A_n\cdot\prod_{1\leq i<j\leq n}\ang{ij}$ are related to each other by momentum conservation and Pl\"ucker relations in the spinors. This is a re-derivation of \cite{Hodges:2012ym} in our new setting.  
Due to its physical interpretations as a scattering amplitude, it is also known that $N_n$ vanishes whenever $\ang{ij}=\sq{ij}=0$ for any $1\leq i<j\leq n$. We give a self-contained proof of this property later on.

\begin{exa}\label{ex:n=5_numerator}
    Let $n=5$, $\mathcal{R}=\{1,2,3\},$ and $\mathcal{C}=\{3,4,5\}.$ Then, 
    \begin{align*}
        \det\Phi^{\mathcal{R}}_{\mathcal{C}}=\det
        \begin{bmatrix}
            \frac{\sq{14}}{\ang{14}} & \frac{\sq{24}}{\ang{24}} \\
            \frac{\sq{15}}{\ang{15}} & \frac{\sq{25}}{\ang{25}}
        \end{bmatrix}
    \end{align*}    
    and thus 
    \begin{align}\label{eq:ex_N_5}
        N_5=\ang{15}\ang{24}\sq{14}\sq{25}-\ang{14}\ang{25}\sq{15}\sq{24}.
    \end{align}
    Notice that the expression does not contain the reference spinors $x,y$. It is also clear that $N_5$ vanishes whenever we set $\ang{ij}=\sq{ij}=0$ for any $(i,j)$ in $\{(1,5),(2,4),(1,4),(2,5)\}$. To make the remaining zeros manifest, we may apply momentum conservation relations as described in Equation \eqref{eq:momconserv}. For example, using the relations
    \[
    \ang{24}\sq{25}=-\ang{14}\sq{15}-\ang{34}\sq{35}\text{  and  }
    \ang{25}\sq{24}=-\ang{15}\sq{14}-\ang{35}\sq{34},
    \]
    we can make the vanishing of $N_n$ at $\ang{34}=\sq{34}=0$ and $\ang{35}=\sq{35}=0$ manifest, as shown by
    \begin{align*}
        N_5 &= \ang{15}\ang{24}\sq{14}\sq{25}-\ang{14}\ang{25}\sq{15}\sq{24} \\
        &= \ang{15}\sq{14}\left(-\ang{14}\sq{15}-\ang{34}\sq{35}\right)-\ang{14}\sq{15}\left(-\ang{15}\sq{14}-\ang{35}\sq{34}\right)\\
        &=\ang{14}\ang{35}\sq{15}\sq{34}-\ang{15}\ang{34}\sq{14}\sq{35}.
    \end{align*}
    Similar substitutions can be made to manifest the remaining zeros. 
    
\end{exa}

\noindent Motivated by the example above we make the following conjecture, which is the main content of our paper:

\begin{conj}\label{conj:main}
The $n$-point MHV gravity amplitude can be written in the form
\begin{equation}
    A_n = \frac{N_n}{\prod_{1\leq i<j\leq n}\langle ij\rangle}
\end{equation}
where the numerator $N_n$ is a polynomial in the brackets $\langle ..\rangle$, $[..]$. In fact, $N_n$ is the {\bf unique polynomial} (up to an overall factor) of bi-degree $(\frac{n^2-3n-6}{2},n-3)$ in $\langle..\rangle$ and $[..]$, that vanishes for any $i<j\in\{1,\ldots,n\}$ if we send a pair of spinor brackets to zero,
\begin{equation}
    \langle ij\rangle = [ij] = 0.
\end{equation}
We will refer to Section \ref{sec:math} for a mathematical concise version of this conjecture.
\end{conj}

From a certain perspective, Hodges formula and our uniqueness conjecture is a generalization of the Parke-Taylor factor (\ref{PT}) for gluon amplitudes to gravitons. In the case of gluons, the MHV tree-level amplitude is fixed by its poles and requirement of logarithmic singularities, and the numerator is trivial. The generalization to arbitrary gluon amplitudes leads us to the positive Grassmannian
and the Amplituhedron. In this framework, all information is contained in the boundary structure, reflected in the poles and singularities of the corresponding canonical differential form. 
It has been known for a long time that gravity amplitudes do not follow this pattern and the singularity structure is more complicated, and yet the amplitude formulas are often remarkably simple \cite{Herrmann:2018dja,Bourjaily:2018omh,Edison:2019ovj,Brown:2022wqr,Belayneh:2024lzq,Cachazo:2024mdn,Bourjaily:2023ycy,Bourjaily:2023uln}. 
Unlike for gluons, the explicit expressions have non-trivial numerators and are not fully fixed only by the locations of their poles. From the physics perspective, this corresponds to ``poles at infinity''. 
Our conjecture suggests that the behavior at infinity is also linked to the strong constraints on the IR region, along the lines of \cite{Cachazo:2005ca,Bedford:2005yy,Herrmann:2018dja,Cachazo:2024mdn}, possibly opening new avenues in the search of the geometric picture for graviton amplitudes.

Our work contributes meaningful computational and theoretical advances that pave the way toward proving this conjecture. The structure of this work is as follows. In Section \ref{sec:math}, we introduce the mathematical framework in which this problem is set and reformulate our main conjecture. Section \ref{sec:steps_proof} shows the steps we believe to be most promising for furnishing a proof of our conjecture in full generality.
Finally, Section \ref{sec:computational_results} then provides explicit computational proofs for small $n$, that is $n=5$ and $n=6$. We conclude with a brief discussion of the shortcomings of our computational methodology for $n\geq 7.$

\section{Spinor-Helcity ideals and little group weights}\label{sec:math}

The goal of this section is to lay out the mathematical framework in which we can formulate Conjecture \ref{conj:main} more succinctly. We do so by expanding on the setup of spinor-helicity varieties given in \cite[Section 2]{spinor-helcity}.
First, let us fix two copies of the complex Grassmannian $\gr(2,n)$, with Pl\"ucker variables $\ang{ij}$ and $\sq{ij}$ for all $i<j\in \{1,\ldots,n\}$, respectively.
We refer to these variables as angle and square \emph{spinors}.
Let $R_n$ denote the polynomial ring in these spinors over $\mathbb{C}$. In this ring we have $n^2$ \emph{momentum conservation} relations, given as
\begin{equation}\label{eq:momconserv}
\sum_{i\in\{1,\ldots,n\}\setminus\{a,b\}}\langle ai\rangle[ib]=0,
\end{equation}
for all $a,b\in \{1,\ldots,n\}$.
Then, denote by $I_n\subset R_n$ the ideal generated by these relations and the quadratic Pl\"ucker relations in the angle and square spinors. We can then identify $I_n$ with the spinor-helicty ideal $I_{n,k,r}$ from \cite[Remark 2.6]{spinor-helcity}, where $k=2$ and $r=0$, defining the corresponding spinor-helicity variety $\operatorname{SH}(2,n,0)$. We will also denote the coordinate ring $\mathbb{C}[\operatorname{SH}(2,n,0)]=R_n/I_n$ by $Q_n$, and the canonical projection $R_n\rightarrow R_n/I_n$, sending elements in $R_n$ to their equivalence class by $\pi_{I_n}$.
In this work, however, we are interested in the following intersection with $\operatorname{SH}(2,n,0)$: Let $J_{ij}\subset R_n$ be the monomial ideal generated by $\langle ij\rangle$ and $[ij]$ for all $i< j\in \{1,\ldots,n\}$, then we want to study a homogeneous part of
\[
J = \bigcap_{1\leq i< j\leq n} \pi_{I_n}(J_{ij})\subset Q_n.
\]
The homogeneous part we are interested in is dictated by a property of the spinor variables that is called \emph{little group weight}, which we are going to introduce next. To that end, we need to understand the spinor variables better. Spinor variables were introduced to make manifest the on-shell condition for massless particles, see \cite{scatteringamplitudes} for details.
More precisely, for a particle labeled by $i$ with complex momentum $p_i$ in $3{+}1$ kinematic space, we want a set of variables which trivialize $
p_i^2=0$ where $p_i^2$ denotes the inner product on $\C^4$.
We can encode each $p_i$ in a complex $2\times 2$ matrix $P_i$. Then, the condition $P_i^2=0$ is satisfied if and only if that matrix is rank deficient, i.e. we can write $\lambda_i\widetilde{\lambda_i}=P_i$ for some suitable $\lambda_i\in\C^{2}$ and $\widetilde{\lambda_i}\in(\C^{2})^\vee$.
The spinor variables $\ang{ij}$ and $\sq{ij}$ are the determinants of the $2\times 2$ matrices given by $\lambda_i\lambda_j$ and $\widetilde{\lambda_i}\widetilde{\lambda_j}$, respectively.
Upon rescaling $\lambda_i$ by a non-zero complex parameter $t_i$, while simultaneously rescaling $\widetilde{\lambda_i}$ by $t_i^{-1}$, the matrix $\lambda_i\widetilde{\lambda_i}=P_i$ remains invariant. This is called \emph{little group scaling} invariance. The MHV numerator $N_n$ is then a function of the spinors that transforms homogeneously under this little group scaling. That is, $N_n(\lambda_1,\ldots,t_i\lambda_i,\ldots,\lambda_n)=t_i^rN_n(\lambda_1,\ldots,\lambda_i,\ldots,\lambda_n)$ for all $i\in\{1,\ldots, n\}$, where $r$ depends on the number $n$ of particles involved in the process. This gives rise to a multi-grading on $R_n$.

For our purposes, a $\mathbb{Z}^m$-graded polynomial ring, is a ring $S=
\C[x_1,\dots,x_n]$ that decomposes as a direct sum of additive groups 
\[
    S=\bigoplus_{v\in \mathbb{Z}^m} S_{v}
\]
such that we get an inclusion
\[
    S_{v}S_{w}\subseteq S_{v+w},
\]
for all $v,w\in \Z^m$.
Therein, the homogeneous parts are denoted by $S_v$, which are the subsets of $S$ consisting of elements of degree $v\in\Z^m$. The degree is induced by a \textit{degree map} $\text{deg}:\mathbb{Z}_{\geq 0}^n\rightarrow \mathbb{Z}^m$, which, for each $u=(u_1,\dots,u_n)$ corresponding to the monomial $x_1^{u_1}x_2^{u_2}\cdots x_n^{u_n}$, assigns $\text{deg}(u)=v\in\Z^m.$
More specifically, in this work, the degree map is given by
\begin{align}\label{eq:LGW}
    w:\mathbb{Z}^{2\binom{n}{2}}_{\geq 0}\rightarrow\Z^{2{+}n},\quad \ang{ij}\mapsto e_{\ang{}}+e_i+e_j,\quad\sq{ij}\mapsto e_{\sq{}}-e_i-e_j,
\end{align}
where $(e_{\ang{}},e_{\sq{}},e_1,\ldots, e_n)$ denotes the ordered standard basis of $\Z^m$ and we implicitly used the identification of the exponent vector $u$ with its corresponding element in $R_n$. 
For example, when $n=5,$ we write
$$w(\ang{13}\ang{24}\sq{45})=(2,1,1,1,1,0,-1)\in \mathbb{Z}^7.$$
In other words, the grading keeps track of the number of appearances of each label $1,\ldots,n$ in $\ang{..}$ or $\sq{..}$, and the numbers of spinor brackets.
This degree map $w$ gives rise to a $\Z^{2{+}n}$ grading on $R_n$, which refines the natural bi-grading with respect to $\ang{..}$ and $\sq{..}$.
Moreover, this grading can be represented as a grading by a $(n{+}2)\times 2\binom{n}{2}$ matrix with non-negative integer entries. This ensures that for all $\alpha\in\Z^{2+n}$ the homogeneous parts $(R_n)_\alpha$ are of finite dimension as $\C$-vector spaces, see e.g. \cite[Section 8.1]{comb_CA}. We say an ideal $J\subset R_n$ is homogeneous if for every element $f\in J$ all of its homogeneous parts are in $J$. For $\alpha\in \Z^{2+n}$ we denote by $(J)_\alpha$ the intersection of $J$ with $(R_n)_\alpha$.
Then, clearly $I_n$ is a homogeneous ideal in $R_n$ with this grading, therefore the quotient $R_n/I_n$ inherits the same grading. Finally, let the \emph{total degree} of $f\in R_n$ be defined as sum of the first two coordinates of $w(f)$; in Physics this is referred to as the \emph{mass dimension} of $f$.
As the numerator $N_n$ of the MHV gravity amplitude $A_n$ is known to carry particular mass dimension and little group weight in all labels $1\leq i\leq n$, we can now encode these physical parameters in the following multidegree:
\[
d(n):=\left(\frac{1}{2}(n^2{-}3n{-}6),n{-}3,n{-}5,\ldots,n{-}5\right)\in \mathbb{Z}^{2+n}.
\]
That is to say, $N_n$ will be a homogeneous polynomial of degree $d(n)\in \mathbb{Z}^{2+n}.$ This formalism allows us to reformulate Conjecture \ref{conj:main}. 

\begin{conj}\label{conj:numerator_unique}
    Fix $n\geq 5$, let $I_n$, $R_n$ and $J$ as above. Then, the $\C$-vector space $J_{d(n)}$ is generated by a unique element, i.e.
    $\dim_{\mathbb{C}}(J_{d(n)})=1$.
\end{conj}

There are a few comments in order.

\begin{rem}\textcolor{white}{.}
\begin{itemize}
    \item 
    In order for this unique generator to be compatible with its physical interpretation as a scattering amplitude, a priori, it would seem appropriate to place an additional constraint on this basis element; namely requiring it to be anti-symmetric under the exchange of any two labels $1\leq i<j\leq n$. However, as we will see this requirement is redundant.
    \item We surmise that the conjecture holds true when coarsening the grading to the natural $\mathbb{Z}^2$-grading which records only the numbers of spinor brackets. We chose to maintain the finer $\mathbb{Z}^{2+n}$-grading, however, in order to expedite the computations carried out in Section \ref{sec:computational_results}.
    \item  
    We can see that Conjecture \ref{conj:main}, is related to Conjecture \ref{conj:numerator_unique} by the correspondence theorem, as the former amounts to saying that $N_n$ is the lowest total degree generator of the ideal $\bigcap_{1\leq i<j\leq n}(J_{ij}+I_n)$ in $R_n$.
    Remarkably, a similar conjecture can be made for the adjoint of polytopes, which is the numerator of their canonical form when considered as positive geometries. Therefore, Conjecture \ref{conj:numerator_unique} constitutes a novel approach to studying the structure of the adjoint hypersurface.
    We also believe that both formulations are equivalent descriptions of the problem, as suggested by our computational results in Section \ref{sec:computational_results}. Since we have both of these formulations of our conjecture, in an effort to streamline notation, we will henceforth refer to the equivalence classes in $Q_n$ as polynomials.
\end{itemize}
\end{rem}

Naturally, we want that $N_n$ is a contender for the unique generator in Conjecture \ref{conj:numerator_unique}. The following Proposition asserts this.

\begin{prop}\label{prop:N_n_basis_element}
    The numerator $N_n$ of the MHV gravity amplitude $A_n$ is a basis element of $J_{d(n)}\subset Q_n$.
\end{prop}

Proving this claim amounts to showing that $N_n$ has a well defined equivalence class in the quotient $Q_n=R_n/I_n$, which is independent of the choices of reference spinors $x,y$ in Equation \eqref{eq:Hodges} and of the sets $\mathcal{R},\mathcal{C}$; and also that $N_n$ is contained in $J_{d(n)}$. We start with a partial result for the former.

\begin{lemma}\label{lem:ref_spinors}
    For any choice of  $x,y\in\{1,\ldots, n\}$ the numerator of $\Phi_{ii}$, as in \eqref{eq:phi_hodges}, has a well defined equivalence class in $Q_n$.
\end{lemma}
\begin{proof}
    For a choice of $x$ and $y$, we write $\Phi_{ii}(x,y)$ to emphasize the dependence of $\Phi_{ii}$ on $x$ and $y$.
    We start by fixing $i=1$, $x=n{-}2$, $y=n{-}1$, and $z=n$. So we need to show that the numerators of $\Phi_{11}(n{-}2,n{-}1)$ and $\Phi_{11}(n{-}2,n)$ are equal in $Q_n$. More precisely, 
    \begin{align}\label{eq:ref_spinors}
        \sum_{k=2}^{n}\sq{1k}\ang{n{-}2k}\ang{n{-}1k}\prod_{j\neq k}\ang{1j}
        =\sum_{k=2}^{n}\sq{1k}\ang{n{-}2k}\ang{nk}\prod_{j\neq k}\ang{1j}
        \in Q_n.
    \end{align}
    First, notice that the summand of the right-hand-side for $k=n{-}1$ is given as 
    \begin{align*}
    &\sq{1n{-}1}\ang{n{-}2n{-}1}\ang{nn{-}1}\prod_{j\neq n{-}1}\ang{1j}\\
    =&
    \ang{nn{-}1}\prod_{j\neq n{-}1}\ang{1j}\left(\sum_{\ell\in\{2,\ldots,n\}\setminus\{n{-}2,n{-}1\}}\sq{1\ell}\ang{\ell n{-}2}\right),
    \end{align*}
    where we used momentum conservation.
    For any fixed $\ell\neq 1,n{-}2,n{-}1$ these summands are 
    \begin{align}\label{eq:example_term}
    \ang{nn{-}1}\prod_{j\neq n{-}1}\ang{1j}\sq{1\ell}\ang{\ell n{-}2}
    =-
    \ang{1\ell}\sq{1\ell}\ang{\ell n{-}2}\ang{n{-}1n}\prod_{j\neq n{-}1,\ell}\ang{1j}.
    \end{align}
    On the other hand, we can use the Pl\"ucker relations for the left-hand-side in Equation \eqref{eq:ref_spinors}, such that for each $k \neq n$ we obtain
    \begin{align*}
    &\quad\prod_{j\neq k}\ang{1j}\sq{1k}\ang{n{-}2k}\ang{n{-}1k}\\
    &=
        -\prod_{j\neq k,n{-}1}\ang{1j}\sq{1k}\ang{n{-}2k}\left(\ang{1k}\ang{n{-}1n}-\ang{1n{-}1}\ang{kn}\right)\\
    &=
        \ang{1k}\sq{1k}\ang{kn{-}2}\ang{n{-}1n}\prod_{j\neq k,n{-}1}\ang{1j}
        -
        \sq{1k}\ang{n{-}2k}\ang{nk}\prod_{j\neq k}\ang{1j}.
    \end{align*}
    Upon inspection, we can see that the first term cancels the term in Equation \eqref{eq:example_term} and the second one cancels a single term in Equation \eqref{eq:ref_spinors} for any $\ell=k\neq n$. Thus, the claim follows for this choice of $x,y$ and $z$. The proof for any other choice is analogous. The overall claim then follows by making two consecutive swaps of reference spinors. 
\end{proof}

Next, we make the observation about the vanishing of $N_n$, when restricting to $\ang{ij}=\sq{ij}=0$ for any $i<j\in\{1,\ldots, n\}$, as alluded to in Example \ref{ex:n=5_numerator}, more rigorous. 

\begin{lemma}\label{lemma:zeros_hodge}
    Fix $x\neq y\in\{1,\ldots,n\}$.
    Then, for all $i<j\in\{1,\ldots,n\}$ such that $\{i,j\}\cap\{x,y\}=\emptyset,$ the numerator $N_n$ of the MHV amplitude $A_n$ equals 
    \[
    N_n=\ang{ij}r_1+\sq{ij}r_2,
    \]
    where $r_1,r_2\in R_n$.
\end{lemma}
\begin{proof}
    Fix $i,j,x$ and $y$ as above. Then, notice that $\sq{ij}$ and $\ang{ij}$ appear only in $\Phi_{ij}=\Phi_{ji}$ and in one of the summands of $\Phi_{ii}$ and $\Phi_{jj}$, more precisely they appear as $\frac{\sq{ij}\ang{xj}\ang{yj}}{\ang{ij}\ang{xi}\ang{yi}}$ and $\frac{\sq{ij}\ang{xi}\ang{yi}}{\ang{ij}\ang{xj}\ang{yj}}$, respectively.
    Notice also that $\ang{ij}$ and $\sq{ij}$ always appear together as $\frac{\sq{ij}}{\ang{ij}}$.
    Now, fix $\mathcal{R}=\mathcal{C}\subset\{1,\ldots,n\}$ of cardinality $3$ and such that $\{i,j\}\cap \mathcal{R}=\emptyset.$ 
    This will ensure $\Phi^{\mathcal{R}}_{\mathcal{C}}$ contains all entries $\Phi_{ii},\Phi_{ij},\Phi_{ji},$ and $\Phi_{jj}$. Relabeling the row and column indices $\{1,\ldots,n\}\setminus \mathcal{R}$ and $\{1,\ldots,n\}\setminus \mathcal{C}$ by $1,\dots,n{-}3$ allows us to write  $\text{det}(\Phi^{\mathcal{R}}_{\mathcal{C}})$ as:
    \[
    \det \Phi^{\mathcal{R}}_{\mathcal{C}} = \sum_{\sigma\in S_{n-3}} \operatorname{sgn}(\sigma)(\Phi^{\mathcal{R}}_{\mathcal{C}})_{1,\sigma(1)}\cdots (\Phi^{\mathcal{R}}_{\mathcal{C}})_{n-3,\sigma(n-3)},
    \]
    where $S_{n-3}$ denotes the symmetric group on $n{-}3$ elements, and $\operatorname{sgn}(\sigma)$ its sign. By our choice of $\{i,j\}\cap \mathcal{R}=\emptyset$ and the observations above, each summand admits a factor $\frac{\sq{ij}^k}{\ang{ij}^k},$ for some $k\in \{0,1,2\}.$ The permutation $\sigma_{i'j'}$, which interchanges the re-labeled $i'$-th and $j'$-th columns containing $\Phi_{ij}$ and $\Phi_{ji}$, will induce a summand of the form $$(-1)\cdot \frac{\sq{ij}^2}{\ang{ij}^2}\cdot \prod_{\ell\in [n]\setminus \mathcal{R}\cup\{i',j'\}}(\Phi^{\mathcal{R}}_{\mathcal{C}})_{\ell,\sigma(\ell)}.$$Hence, the highest power of $\frac{\sq{ij}}{\ang{ij}}$ appearing in $\det(\Phi^{\mathcal{R}}_{\mathcal{C}})$ is two. The only other summand which is quadratic in $\frac{\sq{ij}}{\ang{ij}}$ contains additional factors in the reference spinors, hence the term above cannot be eliminated by any other summands.
    Clearing denominators of  $\det(\Phi^{\mathcal{R}}_{\mathcal{C}})$ will scale each summand by $\ang{ij}^2,$ where $$\ang{ij}^2\cdot \left(\frac{\sq{ij}}{\ang{ij}}\right)^k=\ang{ij}^{2-k}\cdot \sq{ij}^k.$$ Hence, what remains in the numerator is a sum of products each divisible by $\ang{ij}$ or $\sq{ij}$, since either $2-k$ or $k$ for $k\in\{0,1,2\}$ is strictly positive.
\end{proof}

This allows us to prove the main Proposition of this section.

\begin{proof}[Proof of Proposition \ref{prop:N_n_basis_element}]
    This proof is essentially a reformulation of Hodges proof sketch in \cite{Hodges:2012ym}.
    First, note that $d(n)=\deg(N_n)$. Then, using Lemma \ref{lemma:zeros_hodge}, we know that, for a suitable choice of $\mathcal{R}$, $\mathcal{C}$ and the reference spinors $x$ and $y$, the MHV numerator $N_n\in R_n$ can be written as $N_n=\ang{ij}r_1+\sq{ij}r_2$ for all $i<j\in\{1,\ldots,n\}$, thus $\pi_{I_n}(N_n)\in J_{ij}$. By Lemma \ref{lem:ref_spinors}, it remains to show that $N_n$ gives a well defined equivalence class in $Q_n$, independent of the choice of $\mathcal{R}$ and $\mathcal{C}$.
    To that end, we first construct another $n\times n$ matrix $\Psi$ from $\Phi$, by multiplying the $i$-th row with $\ang{1i}\ang{2i}$ for all $1\leq i\leq n$, that is $\Psi_{ij}=\ang{1i}\ang{2i}\Phi_{ij}$. By construction, the column sums of $\Psi$ are zero, as $-\sum_{j\neq i}\Psi_{ij}=\Psi_{ii}$. We also note that the first two rows are zero since $\ang{11}=\ang{22}=0.$
    Next, we want to show that $\det \Psi^{\{1,2,3\}}_\mathcal{C}=-\det \Psi^{\{1,2,4\}}_\mathcal{C}$. To this end, it suffices to show $\det \overline{\Psi}^{\{12\}}_{\mathcal{C}}=0,$ where $\overline{\Psi}$ is the matrix obtained from $\Psi$ by summing the third and fourth rows. But this follows from the fact that 
    \[
    0=\sum_{i=1}^n\Psi_{ij}=\sum_{i=3}^n\Psi_{ij}=\sum_{i=3}^n\overline{\Psi}_{ij},
    \]
    for all $1\leq j\leq n.$ 
    In turn, this implies $\ang{14}\ang{24}\det \Phi^{\{1,2,3\}}_\mathcal{C}=-\ang{13}\ang{23}\Phi^{\{1,2,4\}}_\mathcal{C}$. The proof for any other choice of $\mathcal{R}$ and $\mathcal{C}$, where only one label was swapped is analogous. By stringing multiple of these transpositions of labels together, we have just shown that the numerator $N_n$ is independent of the choice of $\mathcal{R}$ and $\mathcal{C}$. Therefore, we get a well defined element $N_n\in J_{d(n)}\subset Q_n$.   
\end{proof}

\begin{exa}
    Let $n=5$. Then, $I_5$ cuts out the irreducible variety $\operatorname{SH}(2,5,0)$ in $\gr(2,5)\times\gr(2,5)$ with their Pl\"ucker embedding $\mathbb{P}^9\times \mathbb{P}^9$.
    In turn, the intersection with $J$ cuts out a reducible variety of total degree 420 and dimension 8.
    Its irreducible components are precisely given by the vanishing loci of the ideals $J_{ij}\subset R_5/I_5$.
    Moreover, we have $\dim_{\C}(J_{d(5)})=1$ with the generator
    \[
    N_5=\ang{15}\ang{24}\sq{14}\sq{25}-\ang{14}\ang{25}\sq{15}\sq{24}.
    \]
    Notice that this element is the one from Example \ref{ex:n=5_numerator} and it is anti-symmetric upon swapping the labels, a feature that emerges naturally from our description as the intersection of ideals.  
    See Section \ref{sec:computational_results} for a computational derivation of the result.
\end{exa}

\begin{rem}\label{rem:SH_as_flag}
Later on we are going to make use of the parametrization of the spinor-helicity variety $\operatorname{SH}(2,n,0)$, which is also due to \cite{spinor-helcity} and the fact that $\operatorname{SH}(2,n,0)$ is isomorphic to the two step flag variety $\operatorname{Fl}(n,n{-}2;\C^n)$.
This can be seen as follows. First, by definition we have
\[
\operatorname{SH}(2, n, 0) = \{ (V, W) \in \gr(2, n) \times \gr(2, n) \mid \dim(V \cap W^\perp) \geq 2 \},
\]
where $W^\perp$ denotes the orthogonal complement of $W$ in $\C^n$ with respect to the standard inner product.
Then, by passing from $W$ to $W^\perp$ we can identify $\operatorname{SH}(2,n,0)$ with a subvariety in $\gr(2,n)\times\gr(n{-}2,n)$, with points $(V,W^\perp)$ such that $V\subset W^\perp$. 
Moreover, we can describe its Pl\"ucker coordinates as follows. Take an $(n{-}2)\times n$ matrix $X$ with entries given by formal variables $x_{ij}$, that is $X = (x_{ij})$. For, $\{i_1,i_2\}\subset\{1,\ldots, n\}$ the $\ang{i_1i_2}$ Pl\"ucker coordinate of $V$ is the minor of $X$ given by the determinant of the first two rows and the column labeled by $i_1$ and $i_2$.
Similarly, the $\sq{j_1\ldots j_{n{-}2}}$ Pl\"ucker of $W^\perp$ is the determinant of the submatrix of $X$ given by taking the first $n{-}2$ rows and the columns labeled by $\{1,\ldots,n\}\setminus \{j_1\ldots j_{n{-}2}\}$ and multiplying it with $({-}1)^{j_1+\ldots +j_{n{-}2}}$.
Then, we can use \cite[Lemma 3.3]{W_perp_plueckers} (see also \cite{hilbert}) to express this in terms of Pl\"uckers of $W$. That is, if $\{s,t\}$ is the complement of $\{j_1,\ldots, j_{n{-}2}\}$ then 
\[
\sq{st}(W)=({-}1)^{j_1+\ldots +j_{n{-}2}}\sq{j_1\ldots j_{n{-}2}}(W^\perp).
\]
\end{rem}

\section{Special kinematics}\label{sec:steps_proof}

In this section, we outline the steps which we believe will yield a proof of Conjecture \ref{conj:numerator_unique}.
These steps heavily rely on intuition derived from the physical interpretation of $N_n$ as the numerator of the MHV gravity amplitudes and the corresponding factorization properties. We proceed by induction on $n\geq 5,$ where the recursion is established by restricting to special kinematics, which will be defined later in this section. The base cases, $n=5$ and $n=6$ are derived computationally in Section \ref{sec:computational_results}.

We start by endowing $R_n$ with a slightly different $\mathbb{Z}^{2+(n{-}1)}$ grading by defining a degree map 
\[
w':\Z^{2\binom{n}{2}}_{\geq 0}\rightarrow \Z^{2+(n{-}1)},
\]
with $w'(\ang{ij})=w(\ang{ij})$, as in Equation \eqref{eq:LGW}, where, again, we identify the exponent vector $u$ with its corresponding element in $R_n$ for $i,j\in\{1,\ldots, n\}$. And also $w'(\ang{in})=e_{\ang{}}+e_i+e_{n{-}1}$ and $w'(\ang{n{-}1n})=e_{\ang{}}+2e_{n{-}1}$; analogously for the square spinors. Therefore, the grading now counts each appearance of the label $n$ toward the weight in label $n{-}1.$ Notice that this grading is a coarsening of the one introduced in Section \ref{sec:math}. 
To account for this, we define the slightly modified target multidegree $d'(n)$ given by $d'(n)_i=d(n)_i$ for $i\leq n$ and $d'(n)_{n{+}1}=2d(n)_{n{+}2}$. Thus, $N_n\in J_{d'(n)}$. Next, we introduce the notion needed for our induction argument. Denote by $K_{n-1,n}\subset R_n/I_n$ the \emph{special kinematics ideal} at $n$-points given by 
\[
K_{n-1,n}=\left(\ang{n{-}1n}\right)+\sum_{i=1}^{n-1}\left(\ang{in{-}1}-\alpha_{n-1,n}\ang{in}\right),
\]
where $\alpha_{n-1,n}\in \C$ is a non-zero constant. In Physics this correspond to the hard kinematic limit, setting $|n{-}1\rangle=\alpha_{n-1,n}|n\rangle$. Next, define the \emph{special kinematics ring} $Q^K_n=R_n/(I_n+K_{n-1,n})=\pi_K(Q_n)$, where $\pi_K$ is the canonical projection from $Q_n\rightarrow Q_N/K_{n-1,n}$. $Q_n^K$ retains the new $\mathbb{Z}^{2+(n{-}1)}$ grading as $I_n$ and $K_{n-1,n}$ are homogeneous with respect to that grading on $R_n$. Projecting down to the special kinematics ring will be the recursion step in our induction.
Denote by $J^K_{ij}$ the image of the ideal $J_{ij}\subseteq Q_n$ in $Q^K_n$ under $\pi_K$. We then obtain
\begin{align}\label{eq:J_ij_images}
    J^K_{ij}=\begin{cases}
    (\ang{ij},\sq{ij}) & \text{ if }1\leq i<j\leq n-2,\\
    \left(\ang{in{-}1},\sq{in{-}1}\right) & \text{ if }1\leq i\leq n-2,\; j=n-1,\\
    \left(\ang{in{-}1},\sq{in}\right) & \text{ if }1\leq i \leq n-2,\; j=n,\\
    \left(\sq{n{-}1n}\right)& \text{ if }i=n-1,\; j=n.
\end{cases}
\end{align}
We also let $J_n^K:=\bigcap_{1\leq i<j\leq n} J_{ij}^K\subset Q_n^K$, and similarly $J_n=\bigcap_{1\leq i<j\leq n}J_{ij}$ in $Q_n$. With the setup established, we can now state a series of conjectures which together may furnish a proof of Conjecture \ref{conj:numerator_unique}.
We start with an assumption on how the canonical projection $\pi_K$ behaves with respect to taking the intersections of $J_{ij}$.

\begin{conj}\label{conj:can_proj_commutes}
    The canonical projection map $\pi_K$ commutes with taking intersections of the ideals $J_{ij}$ for all $i<j\leq n$. More precisely, 
    \[
    \pi_K(J_n) =\left(\bigcap_{i<j\leq n{-}2}J_{ij}^K\right)\cdot\prod_{1\leq i\leq n{-}2}J^K_{in}\cap J^K_{in{-}1}
    \]
\end{conj}
This result is desirable because it is generally easier to describe elements in a product of ideals rather than in an intersection. 
In fact, if Conjecture \ref{conj:can_proj_commutes} holds, each element in $\pi_K(J_n)$ admits a factor given element in $J^K_{in-1}\cap J^K_{in}.$ This, in turn, allows us to find a designated factor of any polynomial in $f\in\pi_K(J_n)$, as the following Proposition asserts for $5\leq n\leq 7$.

\begin{prop}\label{prop:prefactor}
    Let $5\leq n\leq 7$ and $J_n$ as above. Then, for any element $\widetilde{f}$ in $(J_n)_{d'(n)}$ we have $\pi_K(\widetilde{f})=f$ with
    \[
    f \in (P_n)\cdot \left(\bigcap_{i<j\leq n{-}2}J^K_{i,j}\right)\subset Q_n^K
    \]
    where $P_n=\sq{n{-}1n}\cdot\prod_{1\leq i\leq n-2}\ang{in}\in Q_n^K$.
\end{prop}
\begin{proof} 
    We start with the case $n=7$.
    Set $G=\{g_1,\ldots, g_s\}$ to be a generating set for the ideal 
    \[
    \prod_{1\leq i<6} J_{i7}^K\cap J_{i6}^K.
    \]
    By Conjecture \ref{conj:can_proj_commutes} we then know that any $f\in J_7^K$ can be written as $$f=\left(\sum_{g_i\in G}g_if_i\right)\cdot h$$ for some $f_i\in Q_7^K$ and $h\in \bigcap_{i<j\leq5}J_{ij}^K$. It is clear that every generator $g_k$ is of the form
    \begin{align*}
        g_1=\prod_{1\leq i<6}\ang{i7}, \; \;g_k=\prod_{i\in I}\ang{i7}\prod_{j\in J}\sq{j7}\sq{j6}
    \end{align*}
    for some suitable subsets $I$ and $\emptyset\neq J$ of $[5]$.
    Now consider a summand $g_kf_k\cdot h$.
    It is of homogeneous degree $(11,4,2,2,2,2,2,4)$, which implies that $f_k\cdot h$ must be of degree $(11-|I|,4-2|J|,2,2,2,2,2,4)+e_I-e_J$, where $e_I:= \sum_{i\in I}e_i$ and similarly $e_J:= \sum_{j\in J}e_j.$ 
    Suppose, without loss of generality, that $J=\{1\}$ and $f_k\cdot h\neq 0.$
    Then, $f_k\cdot h$ must be of degree $(7,2,4,2,2,2,2,4)$, which implies that $f_k\cdot h= \ang{17}\cdot r$, for $r\in Q_n^K,$ as there are only $7$ angled spinors on which to distribute the $8$ labels in $1$ and $6$.
    The reasoning for the other cases is analogous. It follows that $$f=\sq{67}\cdot\prod_{1\leq i<6}\ang{i7}\cdot h,$$as claimed.
    The cases for $n=5$ and $n=6$ are analogous, and additionally, with the method described in Section \ref{sec:computational_results}, these can be verified independently.
\end{proof}

We conjecture that such a result holds in general for $n\geq 8.$
\begin{conj}\label{conj:prefactor}
    Let $n\geq 5$, and $J_n\subset Q_n$ as above.
    Then, for any element $\widetilde{f}$ in $(J_n)_{d'(n)}$ we have $\pi_K(\widetilde{f})=f$ with
    \[
    f \in (P_n)\cdot \left(\bigcap_{1\leq i<j\leq n{-}2}J^K_{i,j}\right)\subset Q_n^K,
    \]
      where $P_n=\sq{n{-}1n}\cdot\prod_{1\leq i\leq n-2}\ang{in}\subset Q_n^K$.
\end{conj}

In the recursion step of the induction, there is a need to relate the intersection of ideals in $Q_n$ to $Q_{n{-}1}$. This is established as follows.
\begin{conj}\label{conj:induced_iso}
    There is an induced isomorphism
    \[
    \left(\bigcap_{1\leq i<j\leq n{-}1}J_{ij}\right)_{d(n{-}1)}\cong
    \left(\bigcap_{1\leq i<j\leq n{-}2}J_{ij}^K\right)_{d'(n{-}1)},
    \]
    between vector spaces in $Q_{n{-}1}$ and the special kinematics ring $Q_{n}^K$ respectively.
    It sends $\ang{ij}\mapsto\ang{ij}$ for all $i<j<n$, similarly $\sq{ij}\mapsto\sq{ij}$ for all $i<j\leq n{-}1$  and $\sq{in{-}1}\mapsto \sq{in}+\sq{in{-}1}$.
\end{conj}

As the final ingredient, we need that the lift from the special kinematics ring $Q_n^K$ into $Q_n$ is unique.
\begin{conj}\label{conj:lift_unique}
    The canonical projection $\pi_K:Q_n\rightarrow Q_n^K$ induces a linear isomorphism $\pi_K^*:(J_n)_{d'(n)}\rightarrow (J^K_n)_{d'(n)}$.
\end{conj}

Then, under the assumption of Conjectures 3.1-5
, we can prove Conjecture \ref{conj:numerator_unique} as follows.

\begin{proof}[Proof of Conjecture \ref{conj:numerator_unique}]
By Conjecture \ref{conj:lift_unique} it suffices to show that $\widetilde{N}_n\in Q_n^K$ is unique.
We proceed by induction on $n$.
Let $n=5$.
Then, the claim follows by computation in \texttt{Macaulay2}, see Section \ref{sec:computational_results}.
Next, let $n\geq 5$ be arbitrary but fixed. That is, suppose that $N_{n-1}$ is the unique basis element of $\left(J_{n{-}1}\right)_{d(n{-}1)}\subset Q_{n{-}1}$.
By Conjecture \ref{conj:can_proj_commutes}
any element $\widetilde{f}\in J_n$ with $f=\pi_K(\widetilde{f})$ can be written as $f=p\cdot g,$ where $f\in \prod_{1\leq i\leq n{-}2}J^K_{in}\cap J^K_{in{-}1}$ and $g\in\bigcap_{1\leq i<j\leq n{-}2}J^K_{i,j}$.
If $\widetilde{f}$ is additionally of homogeneous degree $d'(n)$, then by Conjecture \ref{conj:prefactor} we actually have $f=P_n\cdot g$.
Notice, that $w'(f)=d'(n)$ since $\pi_K$ preserves the grading. Moreover, as $P_n$ is of homogeneous degree
\[
w'(P_n)=(n{-}2)e_{\ang{}}+e_{\sq{}}+e_1+\ldots+e_{n{-}2},
\]
we must have $w'(g)=d'(n)-w'(P_n)=d(n{-}1)$. 
Therefore, using the isomorphism of Conjecture \ref{conj:induced_iso} we have $g=N_{n{-}1}$, by the induction assumption. That is to say, $f=P_n\cdot N_{n{-}1}'$ where $N_{n{-}1}'$ denotes the image of $N_{n{-}1}$ under that isomorphism.
Therefore, $f$ is in particular unique and lifts uniquely to $Q_n$ by Conjecture \ref{conj:lift_unique}.
\end{proof}

\begin{exa}
We want to demonstrate explicitly what the proposed steps of the proof look like for $n=6$. Consider the 6-point numerator given by Hodges formula, Equation \eqref{eq:Hodges}, for the choices $\mathcal{R}=\{1,2,5\}$ and $\mathcal{C}=\{3,4,6\}$. Let $S_k(i_1i_2\dots i_k) $ denote the group of permutations of the labels $i_1,i_2,\dots,i_k$ and $|\sigma|$ the sign of the permutation $\sigma\in S_k(i_1i_2\dots i_k)$. We compute:
   \begin{align*}
   N_6
   &=\sum_{\sigma\in S_3(125)} \operatorname{sgn}(\sigma)\sq{13}\sq{24}\sq{56}\ang{23}\ang{35}\ang{14}\ang{45}\ang{16}\ang{26}\\
   &=
   \sq{56}\cdot \sum_{\sigma \in S_2(12)}\operatorname{sgn}(\sigma)\sq{13}\sq{24}\ang{23}\ang{14}\ang{16}\ang{26}\ang{35}\ang{45}\\
   &\quad+\ang{56}\cdot \ang{26}\cdot \sum_{\sigma \in S_2(25)}\operatorname{sgn}(\sigma)\sq{16}\sq{23}\sq{45}\ang{13}\ang{14}\ang{35}\ang{24}\\
   &\quad-\ang{56}\ang{16}\cdot \sum_{\sigma \in S_2(15)}\operatorname{sgn}(\sigma)\sq{26}\sq{13}\sq{45}\ang{23}\ang{35}\ang{14}\ang{24}.
   \end{align*}
   
   Then, on special kinematics, that is taking the image of $N_6$ under the projection map $\pi_K$ we obtain:
   \[\pi_K(N_6)= \sq{56}\ang{16}\ang{26}\ang{35}\ang{45}\big(
   \sq{13}\sq{24}\ang{23}\ang{14}-\sq{14}\sq{23}\ang{13}\ang{24}\big).
   \]
   Firstly, notice that the factor $P_6=\sq{56}\cdot \prod_{1\leq i\leq 4}\ang{i6}$ from Conjecture \ref{conj:prefactor} appears here. We also see that the term in the parenthesis $N_5'$ is the image of $N_5$ under $R_{5}\hookrightarrow R_6\xrightarrow{\pi_{I_6}}Q_6\xrightarrow{\pi_K}Q_6^K$, where, in particular, we use the proposed isomorphism of Conjecture \ref{conj:induced_iso}.
   Notice that we have 
   \begin{align*}
   N_5' = &\sq{13}\sq{24}\ang{23}\ang{14}-\sq{14}\sq{23}\ang{13}\ang{24}\\
   =& \ang{23}\ang{46}\sq{24}\sq{35}-\ang{24}\ang{36}\sq{23}\sq{45}\\
    &+\ang{23}\ang{46}\sq{24}\sq{36}-\ang{24}\ang{36}\sq{23}\sq{46}\\
    =&\ang{23}\ang{46}\sq{24}\left(\sq{35}+\sq{36}\right)-\ang{24}\ang{36}\sq{23}\left(\sq{45}+\sq{46}\right)
   \end{align*}
   in $Q_6^K$ which is clearly contained in the ideal $(\ang{46},\sq{45}+\sq{46})\subset Q_6^K$. This in agreement with the physical phenomenon that, for $|n\rangle=\alpha|n{-}1\rangle$, the image $\widetilde{N_n}$ of $N_n$ under the canonical projection $\pi_K$ exhibits new zeros, corresponding to internal momenta. Put algebraically, we have
    \[ \pi_K(N_6) \in\bigcap_{i\leq 4}\big(\ang{i5},\sq{i6}+\sq{i5}\big)\subset Q_6^K,
   \]
    as expected.
\end{exa}

\section{Computational results}\label{sec:computational_results}

We implement a procedure to compute $(J_n)_{d(n)}$ in \texttt{Macaulay2}, \cite{M2}, for $n=5,6$ in order to deduce Conjecture \ref{conj:numerator_unique} in both cases. 
Let $<$ be the graded reverse lexicographic order on $R_n$, and recall that $I_n$ denotes the spinor-helicity ideal. We compute the standard monomial basis $B_n$ of $(R_n/I_n)_{d(n)}=(Q_n)_{d(n)}$ with respect to $<$ and obtain 
\[
    |B_n|=\begin{cases}
    16 & n=5,\\
    780 & n=6.
    \end{cases}
\]
Next, we take a generic linear combination $$g=\sum_{b\in B_n}c_b b$$in the polynomial ring $\mathbb{C}[\{c_b\mid b\in B_n\}]$ generated by formal variables $c_b$. Denote $\pi_{J_{ij}}:Q_n\rightarrow Q_n/J_{ij}$ the quotient map $g\mapsto g\mod J_{ij}.$ 
Then, $\pi_{J_{ij}}(g)=0$ gives linear conditions in the $c_b$.
Collecting these conditions and organising them into the coefficient matrix with respect to the ordered basis $(c_1,\ldots,c_{|B_n|})$ results in a $(20\times 16)$ matrix $X_5$ and a $(2951\times 780)$ matrix $X_6$, both with entries in $\Z$, for $n=5$ and $n=6$ respectively. Direct computations yield $$\text{dim}(\ker(X_5))=\text{dim}(\ker(X_6))=1.$$ In fact, $\ker(X_5)=N_5$ and $\ker(X_6)=N_6,$ which is consistent with Lemma \ref{lemma:zeros_hodge}, establishing Conjecture \ref{conj:numerator_unique} in both cases. \\

Unfortunately, both generic Gr\"obner basis methods as well as the procedure described above fail for higher $n$.
This is mainly because the combinatorics of the quotient ring $Q_n$ become completely intractable. 
Using the \textit{Mandelstam invariants} $s_{ij}=\langle ij\rangle [ij]$, we were able to find a combinatorial description of bases for $(Q_5)_{d(5)}$ and $(Q_6)_{d(6)},$ but already at $n=7$ such a description becomes infeasible;  a basis for $(Q_7)_{d(7)}$ admits around $10^7$ elements.
In combination with the theoretical framework laid out in Section \ref{sec:math}--which mostly reduces to linear algebra--these impediments starkly highlight the necessity of better tools to deal with such complexity.

\bigskip
\medskip

\noindent {\bf Acknowledgements.} We are grateful to Bernd Sturmfels, Callum Jones and Ben Hollering for the help with the computations. 
J.T., U.O. and S.P. were supported by U.S. Department of Energy grant DE-SC0009999 and funds from the University of California. S.P. is supported by Simons Investigator Award \#376208. J.K. is supported by the European Union (ERC, UNIVERSE PLUS, 101118787). Views and opinions expressed are however those of the authors only and do not
necessarily reflect those of the European Union or the European Research Council
Executive Agency. Neither the European Union nor the granting authority can
be held responsible for them.

\bibliography{References}
\end{document}